\documentclass{svproc}
\usepackage{url}
\usepackage{mathtools}

\begin{document}
\mainmatter              
\title{Performance of Modified Fractional Frequency Reuse algorithm in Random Ultra Dense Networks}
\titlerunning{Performance of Modified FFR algorithm in Random Ultra Dense Networks}  
%
\author{Bach Hung Luu\inst{1}$^{[0009-0009-1908-3882]}$, Samuel Harry Gardner\inst{2}$^{[0009-0003-4049-5861]}$, Sinh Cong Lam \inst{1}$^{[0000-0003-4546-3378]}$ \and
	Trong Minh Hoang\inst{3*}$^{[0000-0001-8486-2940]}$ }
\authorrunning{Bach Hung Luu et al.}
%
\institute{Faculty of Electronics and Telecommunications, VNU University of Engineering and Technology, Hanoi, Vietnam 
	\and 
	Department of Computer Science and Engineering, Sejong University, Seoul, South Korea 
	\and
	Telecommunication Faculty No1, Posts and Telecommunications Institute of Technology, Hanoi, Vietnam
	\email{hoangtrongminh@ptit.edu.vn}}

\maketitle              

\begin{abstract}
Mitigating intercell interference by employing fractional frequency reuse algorithms is one of the important approaches to improving user performance in 5G and Beyond 5G cellular network systems, which typically have a high density of Base Stations (BSs). While most frequency reuse algorithms are based on the downlink Signal-to-Interference-plus-Noise Ratio (SINR) or the distance between the user and its serving BS to classify Cell-Edge Users (CEUs) and Cell-Center Users (CCUs), this paper discusses a modified algorithm that uses the power ratio between the signal strengths from the serving BS and the second nearest BS for user classification. Specifically, if the power ratio is below a predefined threshold, the user is classified as a CEU and is served with higher transmission power. Simulation results show that increasing transmission power is necessary to enhance CEU performance, but it also degrades the performance of typical users. The use of frequency reuse algorithms is particularly feasible in environments with a high density of obstacles, where intercell interference can be effectively suppressed.
\end{abstract}
\section{Introduction}
In the 5G and Beyond 5G cellular system, the BSs tend to be deployed very close together to provide the coverage probability for a massive number of users in complex service areas. Due to the limitation of radio spectrum, the re-utilization of frequency bands between BSs is mandatory to fulfill spectrum requirement and secure the acceptable performance of all BSs as well as their associated users \cite{6Gsurveydensi,ICIC6G}. To mitigate the intercell interference because of spectrum sharing, the Interference Coordination Technique (ICIC) was introduced for the predecessor of cellular systems, such as 4G Long-Term Evolution. By employing ICIC technique, the BSs utilize different power levels to serve their associated users to reduce total power consumption and then reduce intercell interference. 

Research on Inter-Cell Interference Coordination (ICIC) techniques has been conducted across different generations of cellular networks, both with and without advanced enhancements \cite{ICIISAC,ICICAI}. Most existing works on ICIC focus on the Fractional Frequency Reuse (FFR) algorithm, which defines an appropriate reuse pattern for a group of base stations (BSs). Traditionally, this algorithm relies on key criteria to classify users into different groups. In \cite{FFRPhycom,6047548}, the downlink Signal-to-Interference-plus-Noise Ratio (SINR) is used to distinguish users across groups. In other studies \cite{distanceFFR,FFRdistance}, users are grouped based on the distance between the users and their serving BS. In these approaches, users who are farther than a reference distance or experience a high downlink SINR are defined as CCUs. Otherwise, users who are closer or have a relatively lower SINR are categorized as CEUs.

Although SINR-based and distance-based methods are two common approaches for implementing FFR techniques, they possess inherent limitations that need to be addressed. In particular, SINR is a random variable that depends on various factors such as instantaneous channel conditions and the resource allocation status at adjacent base stations. Consequently, the SINR-based approach requires continuous, real-time user classification, which can lead to instability in network operation. Meanwhile, the distance-based approach relies on highly accurate user location estimation, which may be infeasible in certain scenarios. Moreover, near users may occasionally experience lower SINR and require higher transmission power than far users. 

In cellular networks with nearest association schemes, the interference of a typical user mostly originates from the second nearest BS.  Therefore,  the paper proposes a new approach to distinguish CCU and CEU that can overcome the limitations of SINR-based and distance-based approaches. Particularly, the typical user estimates the ratio between the signal strength from its serving BS and the second nearest BS. If the ratio is smaller than the pre-defined threshold, the user is called CCU. In the practical networks, besides the desired signal from the serving BS, the user is always required to estimate the signal from the second nearest BS for advanced purposes such as handover and cooperative communication. Thus, this modified algorithm can be deployed in practical networks. The simulation results for random cellular networks illustrate that:
\begin{itemize}
	\item An increase in transmission power ratio of CCU and CEU improves the CEU performance, but it may result in a decline in coverage probability of the typical user.
	\item Since the impact of intercell interference can be mitigated by the obstacles, the use of a complicated frequency reuse algorithm in the bad channel condition is unnecessary.
\end{itemize}

\section{System model}
In this paper, a cellular network with a high density of base stations (BSs) is studied to evaluate the coverage probability in an indoor environment with the presence of obstacles. To fully cover the network service area, the BSs are assumed to be randomly distributed in two dimensions and are typically modeled using a spatial Poisson Point Process (PPP) with density $\lambda$ (BS/km²). Without loss of generality, the typical user is assumed to be located at the origin and has a connection to the nearest BS. Thus, the joint Probability Density Function (PDF) of the distance from the typical user to its nearest $r_1$ and second nearest BSs $r_2$ is:

\begin{align}
	f = (2\pi\lambda)^2 r_1 r_2 \exp(-\pi\lambda r^2)
	\label{f}
\end{align}

\subsection{Modified Fractional Frequency Reuse}
\label{algorithmdis}

The broadcast channel is affected by interference from adjacent BSs, excluding the nearest and second nearest ones. Let $I$ be the total intercell interference of the typical user on the broadcast channel, the received signal from the nearest and second nearest BSs on this channel are respectively:
\begin{align}
	S_1 = \frac{Pg_1 L(r_1)}{I + \sigma^2} \text{ and } S_2 = \frac{Pg_2 L(r_2)}{I + \sigma^2}
\end{align}
Where $P$ is the transmission power of BSs on the broadcast channel; $L(r_1)$ and $L(r_2)$ are the path loss from the typical user to the nearest and second nearest BSs; $g_1$ and $g_2$ are corresponding instantaneous channel power gains; $\sigma^2$ is the Gaussian power.

Under the Modified Fractional Frequency Reuse discipline, a user is defined as a CEU if the received signal power from the nearest (serving) BS is smaller than a specific signal threshold compared to the signal power received from the second nearest BS. Let $T$ as the signal threshold, then the user is called CEU if
\begin{align}
	\frac{S_1}{S_2} < T \text{ or }\frac{Pg_1 L(r_1)}{I + \sigma^2} / \frac{Pg_2 L(r_2)}{I + \sigma^2} < T
\end{align}
or 
\begin{align}
	\frac{g_1L(r_1)}{g_2L(r_2)}<T
\end{align}

\begin{proposition}
	For a given network scenario, the CCU and CEU classification probability are
	\begin{align}
		p_e = \frac{TL(r_2)}{1+TL(r_2)}
	\end{align}	
	and $$p_c = 1-p_e=\frac{L(r_1)}{1+TL(r_2)}$$.
\end{proposition}

\begin{proof}
	The probability that the user is called CEU or CEU classification probability is 
	\begin{align}
		p_e& = \mathbf{P}\left(\frac{g_1L(r_1)}{g_2L(r_2)}<T\right) \nonumber\\
		&= \mathbf{P}\left(g_1<g_2\frac{TL(r_2)}{L(r_1)}\right)
	\end{align}
	Under assumption that the channel power gain has an negative exponential distribution with the mean of 1, the CEU classification probability is obtained by
	\begin{align}
		p_e = 1-\mathbf{E}\left[\exp\left(-g_2\frac{TL(r_2)}{L(r_1)}\right)\right]
	\end{align}
	Since $g_2$ is also an exponential random variable, $E[e^{-gs}]=\frac{1}{1+s}$. Hence,
	\begin{align}
		p_e = 1-\frac{L(r_1)}{L(r_1)+TL(r_2)} = \frac{TL(r_2)}{L(r_1)+TL(r_2)}
	\end{align}	
	The proposition is proved.
\end{proof}

With the assumption that the BS utilizes a frequency band to serve one user at a specific time slot, the number of BSs that simultaneously transmit on the CEU power level is exactly the number of CEUs. Thus, the density of BSs that simultaneously transmit on CEU and CCU power levels is respectively:
\begin{align}
	\lambda_e = \lambda p_e \text{ and } \lambda_c = \lambda p_c
	\label{BSdens}
\end{align}

\begin{proposition}
	The CEU classification probability over all network scenarios is
	\begin{align}
		\overline{p}_e = \int_{0}^{\infty}\int_{r_1}^{\infty}\frac{TL(r_2)}{L(r_1)+T*L(r_2)}f(r_1,r_2)dr_2dr_1
		\label{overPe}
	\end{align}
	and CCU classification probability is $\overline{p}_c = 1-\overline{p}_e$.
\end{proposition}
\begin{proof}
	The CEU classification probability over all network scenarios is given by the expectation in $p_e$ with respect to $r_1$ and $r_2$ whose joint PDFs are from Equations \ref{f}, with the reminder that $r_2>r_1$, we obtain the CEU classification probability as Equation \ref{overPe}.
\end{proof}

\begin{proposition}
	Average transmission power of BS is used to serve the typical user is
	\begin{align}
		\overline{P} = P(p_c+\phi p_e)
	\end{align}	
\end{proposition}	

Figure \ref{fig:ceudens} shows the density of CEUs that increases consistently with the threshold $T$ and $\beta$, reflecting that more users are classified as CEUs under a worse channel condition. However, the rate of increase and absolute values differ significantly across the three values of $\beta$. At $T = 0$ dB, the CEU density is about 0.45 for $\beta = 0.5$, 0.30 for $\beta = 1.5$, and 0.15 for $\beta = 2.0$, showing a 50\% reduction between each case. At $T = 20$ dB, the gap remains substantial: CEU density reaches over 98\% for $\beta = 0.5$, while staying below 70\% for $\beta = 2.0$. These differences highlight the strong influence of the path loss exponent on user classification. A lower $\beta$ leads to higher interference levels and CEU density, while a higher $\beta$ results in stronger signal isolation and fewer CEUs.

\begin{figure}
	\centering
	\includegraphics[width=0.9\linewidth]{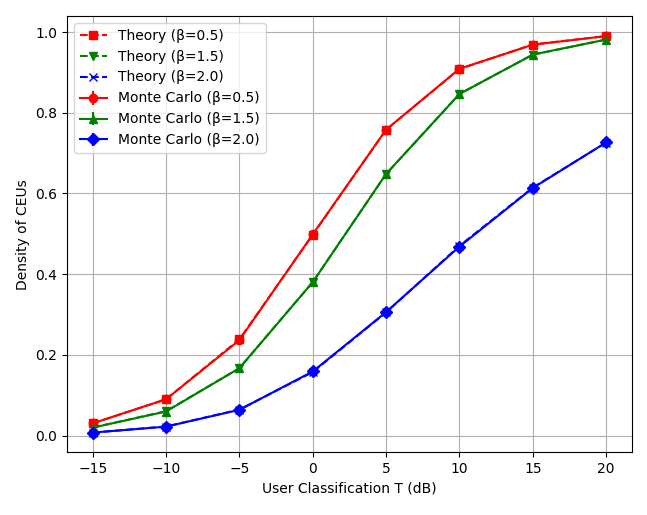}
	\caption{Density of CEUs vs User Classification Threshold T}
	\label{fig:ceudens}
\end{figure}

\subsection{Downlink SIR}
After determining the typical user as either CCU or CEU, the BS selects an appropriate transmission power to serve the typical user. Denote $a$ as the ratio between the serving power of CEU and CCU ($0<a<1$). Specifically, the transmission power of CCU and CEU are $P$ and $ aP$. 
The downlink desired signal of CCU and CEU are
\begin{itemize}
	\item For CCU 
	\begin{align}
		g PL(r_1)|\frac{g_1L(r_1)}{g_2L(r_2)}>T
	\end{align}
	\item For CEU
	\begin{align}
		a g PL(r_1)|\frac{g_1L(r_1)}{g_2L(r_2)}<T
	\end{align}
\end{itemize}

Due to spectrum sharing, the data transmission process between the typical user and its serving BS also experiences interference from other BSs transmitting on the same frequency band, where the densities of BSs transmitting at CCU and CEU power levels are $\lambda p_c$ and $\lambda p_e$, respectively. The total intercell interference of the typical user is
\begin{align}
	I = \sum_{k\in\theta_c}Pg_kL(r_k) + \sum_{k\in\theta_e}Pg_kL(r_k)
	\label{ici}
\end{align}
where $\theta_c$ and $\theta_e$ are set of interfering BSs transmitting on CCU and CEU power levels. 

Consequently, the downlink SIR of the CCU and CEU are written in the following conditional events
\begin{itemize}
	\item For CCU 
	\begin{align}
		SIR_c = \left(\frac{g PL(r_1)}{I}\middle|\frac{g_1L(r_1)}{g_2L(r_2)}>T\right)
		\label{SIRc}
	\end{align}
	\item For CEU
	\begin{align}
		SIR_e = \left(\frac{a g PL(r_1)}{I}\middle|\frac{g_1L(r_1)}{g_2L(r_2)}<T\right)
		\label{SIRe}
	\end{align}
\end{itemize}

\section{Coverage Probability}
In this section, the coverage probability that the typical user is under the BS's coverage area is presented by the conditional probability.
\begin{align}
	\mathcal{P}(\hat{T}) = \mathbf{P}(SIR>\hat{T})
\end{align}
where $\hat{T}$ is called the SIR requirement of the typical user. For CCU and CEU, the SIRs are defined in Equations \ref{SIRc} and \ref{SIRe} respectively. The coverage probability of the typical user is
\begin{align}
	\mathcal{P}(\hat{T}) =& \mathbf{P}(SIR_c>\hat{T})\mathbf{P}\left(\frac{g_1L(r_1)}{g_2L(r_2)}>T\right) \nonumber\\
	&+ \mathbf{P}(SIR_e>\hat{T})\mathbf{P}\left(\frac{g_1L(r_1)}{g_2L(r_2)}<T\right)
\end{align}
By utilizing the Bayes Rules, we obtain
\begin{align}
	\mathcal{P}(\hat{T}) =& \mathbf{P}\left(\frac{g PL(r_1)}{I},\frac{g_1L(r_1)}{g_2L(r_2)}>T\right) \nonumber\\
	&+ \mathbf{P}\left(\frac{a g PL(r_1)}{I},\frac{g_1L(r_1)}{g_2L(r_2)}<T\right)
\end{align}

\section{Performance Analysis}

In this section, the Monte Carlo simulation is derived to examine the dependence of user coverage probability on other network parameters such as transmission conditions, density of BSs, user classification threshold, and so on. The simulation flow chart is summarized as in Figure \ref{fig:simmode}.

\begin{figure}[h]
	\centering
	\includegraphics[width=1\linewidth]{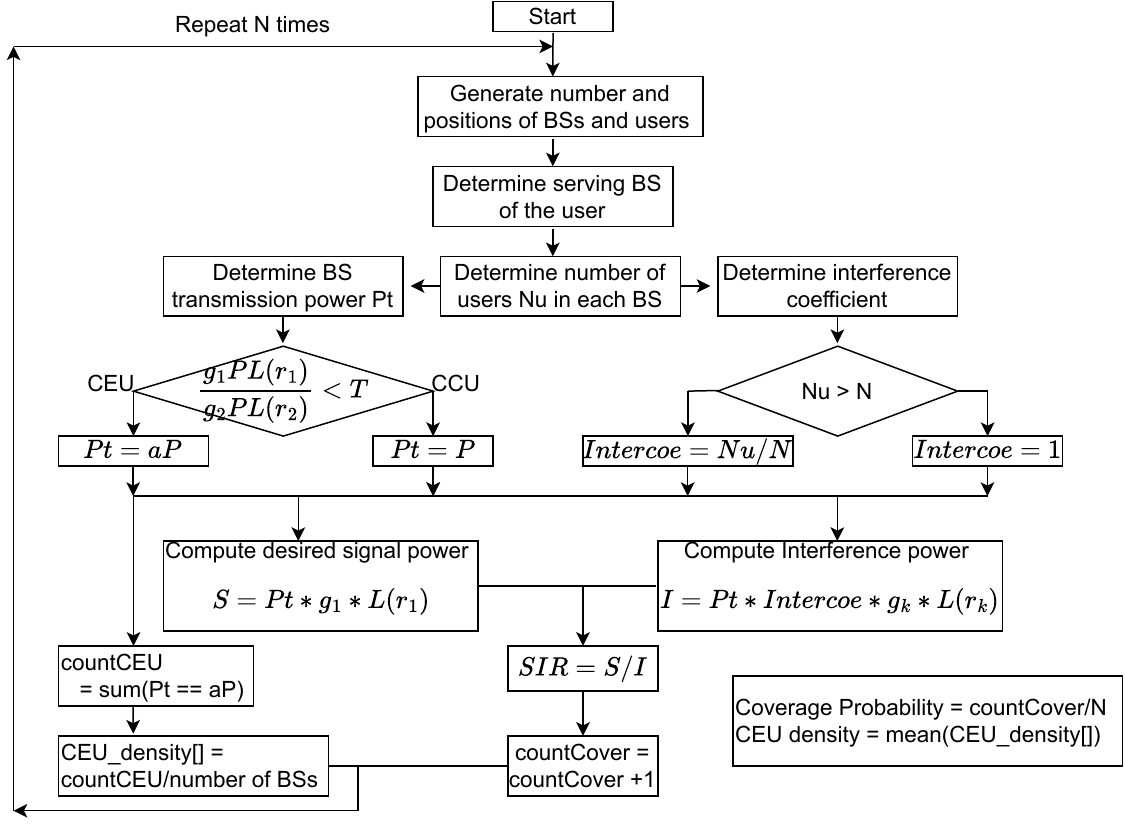}
	\caption{Simulation steps}
	\label{fig:simmode}
\end{figure}

The main steps are described as follows:
\begin{itemize}
	\item At the first steps, the number and position of users and BSs are generated in the 2-D coordinate as the points of two independent PPPs.
	\item After that, the distances between these stations are computed. Based on the nearest association mechanism, the serving BS of each user and then the number of connected users of the BS $N_u$ are determined. 
	\item With the assumption that each frequency sub-band is simultaneously assigned to one user only, if the number of users $N_u$ of the BS is larger than the number of sub-bands $N$, then all sub-bands are occupied and the BS creates interference to the typical user. In contrast, if $N_u<N$, then the BS acts as the interfering source of the typical user with a probability of $\frac{N_u}{N}$. 
	\item To determine the transmission power of the BS on the sub-band of interest, we examine the difference between the received signal strength of the user on that sub-band from its serving BS and its second nearest BS. If this difference is smaller than the classification threshold $T$, then the user is called CCU and the corresponding transmission power on the sub-band is $P$. Otherwise, the transmission power of the BS on that sub-band is $ aP$. Specifically, this step follows the discipline of the frequency reuse algorithm in Section \ref{algorithmdis}.
	\item After determining the Interference Coefficient and BSs' transmission power, the total interfering power of the typical user and then the downlink SIR are computed.  The user coverage probability and density of CEUs are derived by counting the event that $SIR > \hat{T}$ and transmission power of BSs $P_t = aP$.
\end{itemize}

In the following sections, the simulation results are derived for the random cellular network with density of BSs $\lambda = 10^{-2}$ $BS/m^2$ and the density of users $\lambda_u = 10 \lambda$. The number of frequency sub-bands $N=10$; coverage threshold $\hat{T}=-20$ dB. 
\subsection{User Coverage Probability vs Classification Threshold}
Figure \ref{fig:figure1} illustrates the effect of the classification threshold T on user coverage probability. It is worth noting that the nearest association mechanism is employed in the system model, meaning that the user connects to the nearest BS even if the second nearest BS could offer a stronger signal.

\begin{figure}[h]
	\centering
	\includegraphics[width=1\linewidth]{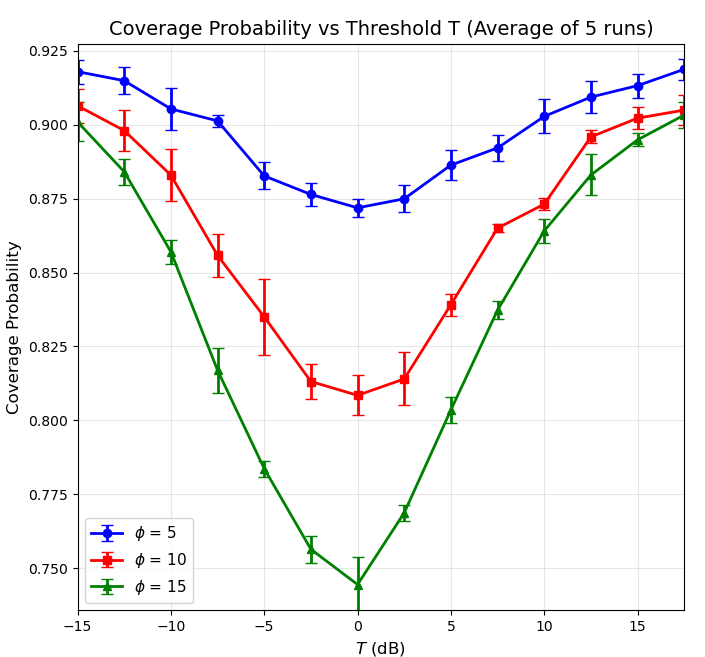}
	\caption{User Coverage Probability vs Classification Threshold}
	\label{fig:figure1}
\end{figure}

As illustrated, the coverage probability exhibits a non-monotonic trend with respect to \( T \), forming a V-shaped curve. When \( T \) increases from \(-15~\mathrm{dB}\) to \(0~\mathrm{dB}\), the coverage probability significantly declines. For instance, when \( a = 10 \), the probability decreases from approximately \(0.905\) at \(T = -15~\mathrm{dB}\) to around \(0.808\) at \(T = 0~\mathrm{dB}\), corresponding to a reduction of nearly \(10.7\%\). This decline results from more users being classified for high-power transmission, increasing not only their signal power but also the overall interference in the network. Since the coverage probability is determined by the signal-to-interference-plus-noise ratio (SIR), an increase in transmission power can be detrimental when the interference rises more rapidly than the desired signal power.

Conversely, as \( T \) continues to increase beyond \(0~\mathrm{dB}\), the coverage probability begins to recover. In the same example with \( a = 10 \), the probability increases from approximately \(0.808\) at \(T = 0~\mathrm{dB}\) to about \(0.902\) at \(T = 15~\mathrm{dB}\), representing a relative gain of roughly \(11.6\%\). This is attributed to the fact that only a smaller subset of users receives high-power transmissions, thereby limiting interference and improving the SIR for others.

Moreover, the effect of the transmission power ratio $a$ is seen in the comparative performance across the three curves. A higher $a$ tends to reduce the coverage probability at the same \( T \) value. For instance, at \( T = 0~\mathrm{dB} \), increasing $a$ from 5 to 15 reduces the coverage probability from approximately \(0.873\) to \(0.744\), which is a substantial drop of around \(14.8\%\). This observation underscores that excessive power scaling may degrade network performance due to elevated interference levels.

\subsection{User Coverage Probability vs Channel Condition}
Figure~\ref{fig:Figure_10} illustrates the variation of the coverage probability as a function of the density of obstacles $\beta$, for three values of obstacle resistance with $ \alpha \in \{0.1, 0.01, 0.001\} $.

\begin{figure}[h]
	\centering
	\includegraphics[width=1\linewidth]{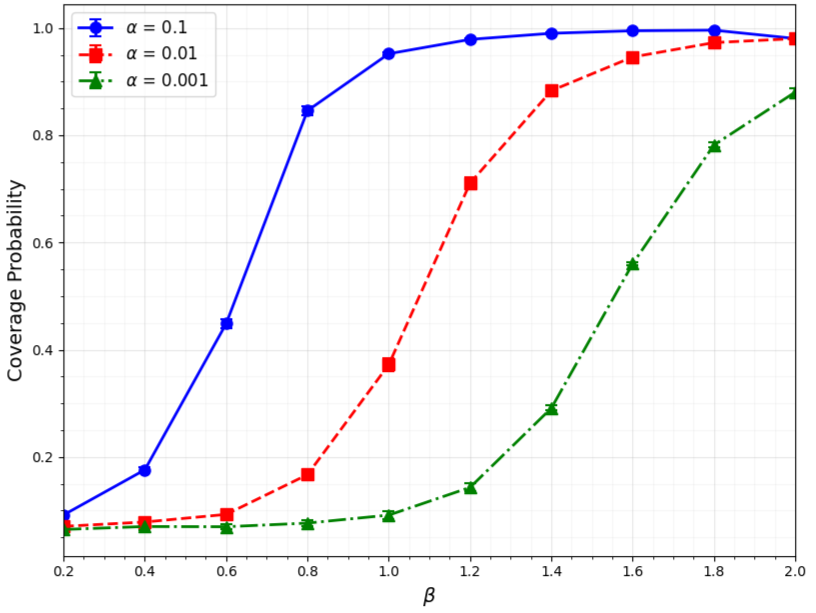}
	\caption{Simulation steps}
	\label{fig:Figure_10}
\end{figure}

Under the stretched path loss model, the path loss increases rapidly with either $\alpha$ or $\beta$, and hence larger values of \( \alpha \) or \( \beta \) imply more aggressive attenuation. This attenuation not only weakens the desired signal but, crucially, also suppresses interference from other adjacent BSs, especially those at greater distances. Therefore, a higher value of $\alpha$ or $\beta$ can result in a higher SIR and consequently a higher coverage probability.

As observed in Figure~\ref{fig:Figure_10}, the coverage probability increases monotonically with \( \beta \) across all considered values of \( \alpha \). For instance, when \( \alpha = 0.1 \), the coverage probability rises sharply from approximately \(0.10\) at \( \beta = 0.2 \) to nearly \(0.99\) at \( \beta = 1.2 \), indicating an increase of nearly \(89\%\). In contrast, for \( \alpha = 0.01 \), the coverage probability reaches a comparable level only at \( \beta \approx 1.9 \).

Additionally, for any fixed \( \beta \), higher values of \( \alpha \) consistently yield better coverage. For example, at \( \beta = 1.2 \), the coverage probability is approximately \(0.99\), \(0.71\), and \(0.15\) for \( \alpha = 0.1 \), \(0.01 \), and \(0.001 \), respectively. This behavior confirms that strong attenuation effectively limits interference, thereby improving SIR and user coverage.

\section{Conclusion}
The paper studies a modified version of Fractional Frequency Reuse for dense networks, in which users are classified based on the power ratio between the signal from the serving base station (BS) and the second nearest BS. The proposed scheme is feasible in practical networks, as users typically observe signals from the second nearest BS for advanced procedures such as handover and cooperative transmission/reception. Simulation results are obtained for random cellular networks with randomly distributed BSs and users. The results illustrate that classifying all users as either CCUs or CEUs can maintain the typical user’s performance, although a higher transmission power ratio can further enhance CEU performance.
\bibliographystyle{splncs03_unsrt} 
\bibliography{refer}

\end{document}